\documentclass{fundam}

\usepackage[ruled,noline,linesnumbered, noend]{algorithm2e}
\usepackage{amssymb,amsmath}
\usepackage{babel}
\usepackage[symbol]{footmisc}
\usepackage{graphicx}
\usepackage{placeins}
\usepackage{subcaption}
\usepackage{tikz}
    \usetikzlibrary{patterns}
\usepackage{url}
\usepackage{xcolor}
\usepackage{hyperref}

\usepackage{float}

\begin{document}
\title{A Fault-Tolerant Version of Safra's Termination Detection Algorithm}

\author{
\!\!Wan Fokkink\\
Department of Computer Science, Vrije Universiteit Amsterdam \and
Georgios Karlos\\
Department of Computer Science, Paderborn University \and
Andy S.\ Tatman\\
Bernoulli Institute, University of Groningen}

\runninghead{W. Fokkink, G. Karlos, A. Tatman}{A Fault-Tolerant Version of Safra's Termination Detection Algorithm}

\maketitle

\begin{abstract}
A distributed algorithm in a message-passing framework has terminated if all nodes in the network have become passive and there are no messages in transit. Safra's distributed termination detection algorithm employs a logical token ring structure within a distributed network; only passive nodes forward the token, and a counter in the token keeps track of the number of sent minus the number of received messages. We adapt this classic algorithm to make it fault-tolerant, meaning that it can cope with node crashes, which are eventually detected by the alive nodes. The counter is split into counters per node, to discard counts from crashed nodes. If a node crashes, the token ring is restored locally and a backup token is sent. Nodes inform each other of detected crashes via the token. Our algorithm imposes no additional message overhead, tolerates any number of crashes as well as simultaneous crashes, and copes with crashes in a decentralized fashion. Correctness proofs are provided of both the original Safra's algorithm and its fault-tolerant variant, as well as a model checking analysis.
\end{abstract}


\section{Introduction}
\label{sec:introduction}

Termination detection, a fundamental problem in distributed systems, was introduced independently in \cite{DijSch80} and \cite{Fra80}. In a message-passing distributed system executing a global computation, termination occurs when two conditions simultaneously hold across the network: every node has completed its local processing and transitioned into a passive state, and all in-flight messages sent prior to passivity have been received and processed. Distributed termination detection is applied in e.g.\ workpools, routing, diffusing computations, self-stabilization, and checking stable system properties such as deadlock and garbage in memory. Many (mostly failure-sensitive) termi\-nation detection algorithms have been proposed in the literature, see \cite{Mattern87,MatCam98}.
Following \cite{Fok18,Tel00}, we call the distributed algorithm for which termination is checked basic and the termination detection algorithm the control algorithm. Likewise, we distinguish basic and control messages.

In Safra's termination detection algorithm \cite{Dij87,FeiGas99} for asynchronous message-passing networks with non-FIFO channels, a designated initiator node, when it has become passive, sends out a white token that visits all nodes in the network via a predetermined logical ring structure; each non-initiator node passes on the token when it is passive. Each node keeps track of the number of outgoing minus incoming basic messages, and these counts are accumulated in the token. Nodes that receive a basic message are colored black, as the count in the token may be unreliable, if the message overtook the token in the ring, because its receipt is recorded in the token but not yet its send. The black color is transferred to the token at its next visit. If the token returns to the initiator white and with counter $0$, the initiator can announce termination. Else the initiator sends out a token again when it is passive.

Safra's algorithm imposes only little control message overhead when nodes remain active over a long period of time, unlike termination detection algorithms for which every control message needs to be acknowledged (e.g.\ \cite{DijSch80}). Additionally, it does not require control messages to be sent out when nodes become passive and does not run into underflow issues, as opposed to weight-throwing schemes (\cite{Mat89,Tse95}). In \cite{DemAro00} an optimized version of Safra's algorithm was proposed, that does not always color receiving nodes black and detects termination within a single round trip of the token after actual termination has occurred.

We propose a fault-tolerant termination detection algorithm based on the improved version of Safra's algorithm from \cite{DemAro00}, meaning that it can cope with nodes crashing permanently. In this setting, termination is defined as follows \cite{Fok18}: An execution of a basic algorithm has terminated if all alive nodes are passive and for all basic messages in transit, the destination node either crashed or knows that the sender crashed. In our algorithm a node crash is handled locally by its predecessor in the ring; a new token is issued, as the old token may have been lost in the crash. A numbering scheme in the token makes sure only a single token is being passed on; if the old token was not lost in the crash, the new token will be dismissed. Only the basic message exchange between alive nodes is counted. For this purpose each alive node keeps track how many basic messages it sent to minus received from each other node individually, so that such counts can be discarded if another node crashes. Likewise, the counter in the token is split into one counter per node. Each node has a failure detector and informs other nodes of crashes reported by its failure detector via the token, so that the alive nodes uniformly count the basic message exchange with the same set of alive nodes. An alive node reporting a crashed node ensures that the token completes another round trip, to avoid inconsistent message counts in the token.

Next to the aforementioned strong points of Safra's algorithm, our fault-tolerant variant has some additional advantages, compared to existing fault-tolerant termination detection algorithms which will be discussed in Section \ref{sec:related}.
First, the basic algorithm can be decentralized, meaning there can be multiple initiator nodes. Second, if the initiator of our termination detection algorithm crashes before sending out the first token, this role is automatically taken over by its predecessor. Thus our algorithm can cope with any number of simultaneous node crashes. Third, only one additional token message is required for each crash, next to possibly an additional round trip of the token. No relatively expensive schemes like leader election or taking a global snapshot are employed. The price to pay is that, in the absence of stable storage, the bit complexity of a message is $\Theta(N)$ with $N$ the number of nodes in the network, compared to $\Theta(1)$ for the failure-sensitive version of Safra's algorithm. Considering current network technologies, with Gbits/second throughput and microseconds latency, this token size incurs a tolerable overhead in network load, especially since the token is only forwarded by idle nodes.

We provide correctness proofs of both the failure-sensitive and the fault-tolerant version of Safra's algorithm. Moreover, an extensive model checking exercise was performed using the mCRL2 tool, revealing a subtle bug in the original fault-tolerant version of Safra's algorithm presented in \cite{KFF21}. A simple fix is proposed for this flaw.

An earlier version of this paper appeared as \cite{KFF21}. Compared to that paper, the main additions are the correctness proofs and the model checking analysis. Developing our algorithm was a delicate matter. Still, owing to the correctness proofs in combination with the model checking analysis, we can confidently claim that our algorithm correctly detects termination. In \cite{KFF21} moreover experimental results in a multi-threaded emulation environment were presented. Our algorithm was tested on two fault-tolerant distributed algorithms from the literature, for networks up to 144 nodes. Compared to the failure-sensitive version of Safra's algorithm, the fault-tolerant version turned out to exhibit a satisfactory performance, in the sense that it imposes no additional control message overhead. Of course it does impose some overhead, by adding extra concurrency at each node and additional synchronization. However, even with a large number of failures, profiling of the experiments showed that the execution time of the basic algorithm remains the dominant factor for the overall performance.

The paper is structured as follows. Section \ref{sec:related} discusses related work and Section \ref{sec:model} the system model. In Section \ref{sec:safra} the improved version of Safra's termination detection algorithm is presented, and in Section \ref{sec:fault-tolerant} our fault-tolerant version, both accompanied with pseudocode descriptions and correctness proofs. Our model checking analysis is discussed in Section \ref{sec:mcrl2}. The paper is concluded in Section \ref{sec:conclusion}.
\section{Related Work}\label{sec:related}

We discuss some existing fault-tolerant termination detection algorithms, mainly from a functional point of view. Generally, a \emph{perfect} failure detector \cite{ChandraToueg96} is required to solve termination detection in the presence of crashes \cite{HHMRT00,MFVP08}. Such a failure detector never falsely suspects that a node crashed and eventually detects each node crash.

Lai and Wu \cite{LaiWu95} proposed a fault-tolerant variant of the Dijkstra-Scholten algorithm \cite{DijSch80} for centralized basic algorithms, meaning there is a single initiator node. 
Active nodes  are in the tree, rooted in the initiator, which announces termina\-tion when the tree has disappeared. In the event of a crash, all alive nodes communicate with the designated root node, causing a sequential bottleneck. 

Lifflander et al.\ \cite{LMK13} presented a series of algorithms based on \cite{DijSch80} that avoid the bottleneck of \cite{LaiWu95}. These algorithms are resistant to single-node failures but are only probabilistically tolerant to multi-node failures and incur additional control messages even in crash-free executions. In case of a crash the tree is reconstructed locally. If two nodes fail concurrently, the algorithms may not be able to recover. The algorithm then detects that this is the case. Failure of the root node cannot be handled. Performance results are reported based on an experimental setup, consisting of three mock-up parallel algorithm implementations. Results show acceptable processing time overhead. Message overhead results are not reported.

Tseng \cite{Tse95} developed a fault-tolerant variant of weight-throwing \cite{Mat89} for centralized basic algorithms.
Nodes donate part of their weight to the basic messages they send. The receiver claims that weight on receipt. A node that becomes passive returns its weight to the leader,
who announces termination when it is passive and reclaimed its original weight. The number of control messages increases linearly with the number of basic messages. The algorithm is vulnerable to underflow of weight values and control messages require space to represent floats at high precision.
A global snapshot is taken when a new crashed node is detected. When the leader crashes, an election scheme is employed.

Venkatesan's algorithm \cite{Ven89} turns a failure-sensitive termination detection based on a spanning tree, with the leader node at the root in charge of announcing termination, into a fault-tolerant one. When a crash occurs, a global snapshot is taken. If the leader crashes, an election is held. To cope with $k$ simultaneous crashes, there must be $k$ backup leaders. The local stacks at the nodes must be continuously replicated by the leader and its backups. Upon learning of a crashed node, the leader and its backups simulate the local stacks of the nodes in the system.

Hursey and Graham \cite{HurGra11} developed a termination detection scheme for their fault-tolerant ring-based MPI application. Their algorithm relies on a leader election scheme and fault-tolerant primitives provided by MPI.

Mittal et al.\ \cite{MFVP08} introduced a general framework for transforming any failure-sensitive termination detection algorithm into a fault-tolerant variant that can cope with any number of node crashes. The basic idea is to restart termination detection after each node crash. When applied to existing failure-sensitive algo\-rithms, the resulting fault-tolerant algorithms 
have a significant overhead in control messages, even when no nodes become passive or crash.

In \cite{KFF21}, testing the behavior of fault-tolerant distributed algorithms on very large networks turned out to be challenging. Emulating basic algorithms by means of unrestricted randomization results in executions that refused to terminate on large networks and  do not faithfully mimic all aspects of real-life distributed algorithms. Moreover, in experiments on top of two actual algorithms, allocating multiple network nodes on a single compute node influences the results. These challenges may partly explain why \cite{LMK13} is the only related paper we are aware of to report experimental results, for networks of up to 2048 nodes.

Finally we mention two model checking analyses of termination detection algorithms. In \cite{KKM22} the non-optimised version of Safra's algorithm is analyzed using TLA$^+$. In \cite{HAN18} the fault-tolerant termination detection algorithm from \cite{Tse95} is analyzed using {\sc Uppaal}.
\section{System Model}
\label{sec:model}
We assume a fully asynchronous message-passing system with no shared memory or global clock. Messages may arrive in any order and delays are unbounded but finite. The $N$ nodes are \emph{logically} organized in a ring and are assigned unique, totally-ordered IDs.
Failures are permanent; once a node crashes, it halts and never recovers.
Like most fault-tolerant distributed algorithms in the literature, we require that the underlying physical network provides a reliable bidirectional communication channel between each pair of nodes.

Nodes are either \textit{active} or \textit{passive}. An active node can  send/receive basic messages, perform internal events, or become passive when it \textit{terminates locally}. A passive node cannot send basic messages or perform internal events and only becomes active upon receipt of a basic message. Termination may be announced while  basic messages from crashed nodes are in transit. It is therefore required that passive nodes never become active by the receipt of a basic message sent by a node they know has crashed. An execution of the basic algorithm has terminated if all alive nodes are passive and for all basic messages in transit, the destination node either has crashed or knows that the sender has crashed. Any algorithm for the termination detection problem must satisfy the following two properties. {\it Liveness}: if the system has terminated, this is eventually detected by an alive node; and {\it Safety}: when termination is detected, the system terminated at some point in the past.

In case nodes may crash, we assume a complete network topology and the presence of a perfect failure detector \cite{ChandraToueg96}. Such a failure detector, which never falsely suspects that a node crashed and eventually detects each node crash, can be built if there are known upper bounds on network latency and on process execution times, as well as bounded clock drift to accurately measure those bounds. A node is said to know that another node has crashed if its failure detector has reported this. Following \cite{Fok18}, we define that an execution of a basic algorithm has terminated if all alive nodes are passive and for all basic messages in transit, the destination node either crashed or knows that the sender crashed.
\section{Safra's Termination Detection Algorithm}
\label{sec:safra}

Safra's (failure-sensitive) termination detection algorithm \cite{Dij87,FeiGas99} generalizes the Dijkstra-Feijen-van Gasteren algorithm \cite{DFG83} from synchronous to asynchronous message passing networks. We give a detailed description of Safra's algorithm, including improvements from \cite{DemAro00}. This will serve as a basis for the description of our fault-tolerant version later on.

Safra's algorithm is centralized, with node $0$ as initiator. The basic algorithm to which Safra's algorithm is applied however is allowed to be decentralized and does not need to be ring-based. A token $t$ circulates the ring, starting at the initiator of the control algorithm when it becomes passive for the first time, and being forwarded by the other nodes once they are passive. The field $\textit{count}_{t}$ in $t$ represents the number of basic messages in transit during the round trip of $t$. Each node $i$ records in $\textit{count}_{i}$ the number of basic messages it sent minus the number of basic messages it received, since the last time it forwarded the token. Each time $t$ is received by a node $i$, $\textit{count}_{i}$ is added to $\textit{count}_{t}$, and  $\textit{count}_{i}$ is reset to zero. Upon return of $t$ to the initiator, after it has become passive, termination is detected if $\textit{count}_{t}$ is zero.

The token can during its round trip underestimate the number of basic messages in transit, if the receipt of a message is accounted for in the token before the send of this message. To recognize this, colors black and white are used. Initially all nodes are white, and when the initiator sends out a fresh token, the token is white. When a node $i$ receives a basic message $m$, it may be that the send of $m$ was not yet recorded in $\textit{count}_{t}$. Due to this uncertainty, termination cannot be announced at the end of the upcoming token round. Therefore, upon receipt of $m$, $i$ marks itself black. When $t$ visits a black node, $t$ becomes black and the node white; from then on $t$ remains black for the rest of the round trip. When the initiator has received back $t$, has become passive, and has added the value of $\textit{count}_{0}$ to $\textit{count}_{t}$, it decides whether termination can be detected. If $t$ is black or $\textit{count}_{t}$ is not zero, the initiator sends out a fresh white token again. Otherwise, it can safely announce that the execution of the basic algorithm has terminated.

Two enhancements of Safra's algorithm to reduce detection delay were given in \cite{DemAro00}. One inefficiency is that a basic message always blackens its receiver. Actually an inconsistent count in the token only occurs when the token records the receipt of a message before its send. This happens when a message overtakes the token, meaning that it is sent after the sender was visited by the token, but reaches the receiver before it is visited by the token. A second inefficiency of the original algorithm is that termination is detected only at the initiator. One enhancement allows detection to occur at any node. When multiple nodes can detect termination, there is a second situation in which an inconsistent count in the token can occur. When both the sender and the receiver of a message are ahead of the token, but the receiver will be visited by the token before the sender, it is possible for the receiver to detect termination before the sender is visited.

To deal with the aforementioned problematic scenarios, a sequence number $\textit{seq}_{i}$ is intro\-duced at every node $i$, starting at zero. When a node forwards the token, it increases its sequence number by one, so that nodes in the visited region have a higher sequence number from those in the unvisited region. A node piggybacks its sequence number $\textit{seq}_{m}$ to every basic message $m$ it sends. Using the sequence number, an offending message can be detected if it has a higher sequence number than the receiver, or they both have the same sequence number but the sender has a higher ID than the receiver. Since multiple nodes can detect termination, an offending message should not only blacken the receiver but also all subsequent nodes in the ring up to (but not including) the sender. At all these nodes, the token holds an inconsistent count. So none should detect termination.

The field $\textit{black}_{t}$ in the token is now a node ID, expressing that all nodes the token visits from now up to (but not including) $\textit{black}_{t}$ are black. When the token is sent by the initiator for the first time, $\textit{black}_{t}=N-1$, so that all nodes from $1$ up to $N-1$ are initially considered black. Hence termination can only be detected after the token has visited all nodes at least once. Likewise, $\textit{black}_{i}$ at a node $i$ represents that all nodes that the token visits from $i$ up to $\textit{black}_{i}$ are black. Initially $\textit{black}_{i}=i$ at all nodes $i$, meaning that $i$ considers all nodes white. If a node $i$ receives a basic message $m$ of which the send may not have been accounted for in the token, then $\textit{black}_{i}$ is set to the furthest node from $i$ among $\textit{black}_{i}$ and the sender of $m$. The function $\textit{furthest}_{i}(j,k)$ computes whether node $j$ or $k$ is furthest away from $i$ in the ring. It is defined by:\vspace*{-.5mm}
\[
\begin{array}{lcl}
\textit{furthest}_{i}(j,k) &=& \left\{
\begin{array}{ll}k &\mbox{if \,$i\leq j\leq k$\, or \,$k<i\leq j$\, or \,$j\leq k<i$}\vspace{.5mm}\\
j &\mbox{otherwise}
\end{array}\right.
\end{array}
\]

\vspace*{-1.5mm}

If the token reaches a node $i$, it must wait until $i$ is passive. Then $i$ adds the value of $\textit{count}_{i}$ to the value of $\textit{count}_{t}$. If $i$ is white, meaning that $\textit{black}_{t}=\textit{black}_{i}=i$, it can determine termination in the same way the initiator does in Safra's algorithm: check whether the value of $\textit{count}_{t}$ is zero. If $i$ is black or detects no termination, it forwards the token to its successor. Before doing so, it sets $\textit{black}_{t}$ to $\textit{furthest}_{i}(\textit{black}_{t},\textit{black}_{i})$ if this value is not $i$, or else to $(i+1)\bmod N$. The latter means the successor of $i$ in the ring will consider the token white. Finally, $i$ sets $\textit{count}_{i}$ to zero and $\textit{black}_{i}$ to $i$ and increases $\textit{seq}_{i}$ by one.

Algorithms \hyperref[algo:ifs]{1-4} present the pseudocode of the four procedures available at each node $i$ for the improved version of Safra's algorithm: initialization, sending/receiving a basic message $m$ to/from a node $j$ (SBM/RBM), and receiving a token (RT). Subscript $i$ of a procedure name represents the node where the procedure is performed. For variables, subscript $i$ means a local variable at node $i$ and subscript $t$ a variable in the token. Action $\textit{send}(m,j)$ denotes that message $m$ is sent to node $j$, and the Boolean field $\textit{passive}_{i}$ is $\textit{true}$ only when node $i$ is passive. In line 6 of RT, in the token $t$ that is sent, the values of $\textit{count}_t$ and $\textit{black}_t$ have been adapted, compared to the $t$ in the argument of RT. Within a single node, procedures are generally executed without interruption. The only exceptions to this may occur when nodes are waiting to become passive, in line 3 of Initialization and line 1 of RT, where a node is allowed to perform SBM and RBM calls.

\begin{minipage}{0.85\textwidth}
\begin{algorithm}[H]
    \label{algo:ifs}
    $\textit{count}_{i} \leftarrow 0$;
    ~~$\textit{black}_{i} \leftarrow i$;
    ~~$\textit{seq}_{i} \leftarrow 0$\;
    \If{$i=0$}
    {
        $\textit{wait}(\textit{passive}_{0})$\; \label{algoline:WaitInit}
        $\textit{count}_{t}\leftarrow \textit{count}_{0}$;
        ~~$\textit{black}_{t}\leftarrow N-1$;
        ~~$\textit{send}(t,1)$;
        ~~$\textit{count}_{0}\leftarrow 0$;
        ~~$\textit{seq}_{0} \leftarrow 1$\;
    }
    \caption{$\mbox{\rm Initialization}_{i}$}
\end{algorithm}
\end{minipage}

\begin{minipage}{0.85\textwidth}
\begin{algorithm}[H]
    \caption{$\mbox{\rm SendBasicMessage}_{i}(m,j)$~}
    $\textit{seq}_{m} \leftarrow \textit{seq}_{i}$;
    ~~$\textit{send}(m,j)$;
    ~~$\textit{count}_{i} \leftarrow \textit{count}_{i} + 1$\;
\end{algorithm}
\end{minipage}

\begin{minipage}{0.85\textwidth}
\begin{algorithm}[H]
\caption{$\mbox{\rm ReceiveBasicMessage}_{i}(m,j)$~}{}{}
    \If{$
    \textit{seq}_{m} = \textit{seq}_{i}+1
    \vee (j > i \wedge\,\textit{seq}_{m} = \textit{seq}_{i})$}
        {$\textit{black}_{i}\leftarrow\textit{furthest}_{i}(\textit{black}_{i},j)$\;}
    $\textit{count}_{i} \leftarrow  \textit{count}_{i} - 1$\;
\end{algorithm}
\end{minipage}

\begin{minipage}{0.85\textwidth}
\begin{algorithm}[H]
    \caption{$\mbox{\rm ReceiveToken}_{i} (t)$~}
    $\textit{wait}(\textit{passive}_{i})$\; \label{algoline:WaitRT}
    $\textit{count}_{t} \leftarrow \textit{count}_{t}+\textit{count}_{i}$;
    ~~$\textit{black}_{i}\leftarrow \textit{furthest}_{i}(\textit{black}_{i},\textit{black}_{t})$\;
    \If{$\textit{count}_{t}=0\wedge\textit{black}_{i}=i$}
    {
         {\rm Announce termination}\;
    }
    $\textit{black}_{t}\leftarrow\textit{furthest}_{i}(\textit{black}_{i},(i+1)\bmod N)$\;
    $\textit{send}(t, (i+1)\bmod N)$\; \label{algoline:SendRTSensi}

    $\textit{count}_{i} \leftarrow 0$;
    ~~~~$\textit{black}_{i} \leftarrow i$;
    ~~~~$\textit{seq}_{i} \leftarrow  \textit{seq}_{i}+1$\;
\end{algorithm}
\end{minipage}

\begin{example}
 We consider one possible run of the improved version of Safra's algorithm on a ring of three nodes in Figure \ref{example:fs}, with $\textit{count}_{i}=\textit{seq}_{i}=0$ and $\textit{black}_{i}=i$ at all nodes $i$. 
Initially (left), only node $2$ is active; it sends a basic message $m$ to node $1$, with its ID $2$ and its $\textit{seq}_{2}$ value $0$ attached, sets $\textit{count}_{2}$ to $1$, and becomes passive. The initiator $0$ sends out a token with $\textit{count}_{t}=0$ and $\textit{black}_{t}=2$, which travels via node $1$ to node $2$, where $\textit{count}_{t}$ is set to $1$. On the way, $\textit{seq}_{0}$, $\textit{seq}_{1}$ and $\textit{seq}_{2}$ are set to $1$ and $\textit{count}_{2}$ is set to $0$. At nodes $1$ and $2$ no termination is detected because $\textit{black}_{t}=2$ and $\textit{count}_{t}=1$, respectively. Node $2$ sends the token to the initiator with $\textit{count}_{t}=1$ and $\textit{black}_{t}=0$, leading to the following picture. Node IDs are placed within the nodes; white nodes are passive and gray nodes (in the picture further below) are active; the triple at each node represents the values of $\textit{count}$, $\textit{black}$ and $\textit{seq}$, respectively, at this node.

\begin{figure}[ht]

\begin{subfigure}{0.45\textwidth}
    \centering
    \scalebox{.9}{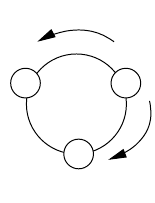}
\end{subfigure}
\begin{subfigure}{.45\textwidth}
    \centering
    \scalebox{.9}{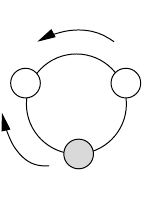}
\end{subfigure}
\caption{Example run on a fault-free network of 3 nodes.}
\label{example:fs}
\end{figure}
Next, $m$ reaches node $1$, making it active and setting $\textit{count}_{1}$ to $-1$; the value of $\textit{black}_{1}$ remains $1$ because the sender of $m$ is $2>1$ and $\textit{seq}_{m}=0\neq \textit{seq}_{1}$. 
Node $1$ sends a basic message $m'$ to the initiator (right), with its ID $1$ and its $\textit{seq}_{1}$ value $1$ attached, and sets $\textit{count}_{1}$ to $0$. $m'$ reaches the initiator before the token, making it active and setting $\textit{count}_{0}$ to $-1$. Moreover, $\textit{black}_{0}$ is set to $1$, since the sender of $m'$ is $1>0$ and $\textit{seq}_{m'}=1=\textit{seq}_{0}$. Now the initiator becomes passive again. The token reaches the initiator, where $\textit{count}_{t}$ and $\textit{count}_{0}$ are both set to $0$. No termination is detected because $\textit{black}_{0}=1$. The token travels on to node $1$ with $\textit{black}_{t}=1$. Next, this node becomes passive and announces termination because $\textit{count}_{t}=0$ and $\textit{black}_{1}=\textit{black}_{t}=1$.

\end{example}

At some points we deviated from the algorithm in \cite{DemAro00}. There, the second disjunct in line 2 of the pseudocode of procedure RBM is omitted, allowing for premature termination detection. Their approach is limited to centralized basic algorithms; in case of a decentralized algorithm, premature termination detection would be possible before the token has completed one round trip, because initially all nodes are white. They interpret sequence numbers modulo 2, so that a receiving node is spuriously blackened in situations where a basic message is overtaken by the token. We condensed $N$ Boolean fields into a single node index \textit{black} at the nodes and in the token, and $N$ integer fields in the token into a single counter.

We argue that the improved version of Safra's algorithm correctly detects termination, since we modified the algorithm in \cite{DemAro00}. The proof sketch will be a corner stone for the correctness argument of our fault-tolerant algorithm in the next section.

\begin{theorem}The improved Safra's algorithm is a correct termination detection algorithm\end{theorem}

\begin{proof}
First we argue {\it Liveness}: if the basic algorithm has terminated, then eventually termination will be announced. Since all nodes are passive, they all proceed beyond $\textit{wait}(\textit{passive}_{i})$ in line 1 of RT$_i$ and pass on $t$. The counter of $t$ accumulates all changes in the counters at the nodes (by line 2 and $\textit{count}_{i} \leftarrow 0$ in line 7 of RT$_i$). Since the sender of a basic message increases its counter by $1$ (line 1 of SBM$_i$) and the receiver decreases its counter by $1$ (line 3 of RBM$_i$), and there are no basic messages in flight, $\textit{count}_{t}$ eventually (and permanently) becomes $0$. Furthermore, since RBM$_i$ will not be called anymore at any node, $\textit{black}_{i}$ and $\textit{black}_{t}$ are only mutated in lines 2,7 and line 5 of RT$_i$, respectively. Therefore eventually $\textit{black}_{i}$ permanently becomes $i$, while $t$ always reaches a node $i$ with $\textit{black}_{t}=i$. Hence eventually $t$ reaches a node $i$ with $\textit{count}_{t}=0$ and $\textit{black}_{i}=i$, so that termination is announced (lines 3-4 of RT$_i$).

Now we argue {\it Safety}: if termination is announced, then the basic algorithm has terminated. Let some node $i$ reach line 3 of RT$_i$ with $\textit{black}_{i}=i$. We argue that then $\textit{count}_{t}\geq 0$. Namely, the only way this value could be negative at this point is if some basic message $m$ was accounted for in $\textit{count}_{t}$ at its receiver $k$, but not at its sender $j$, meaning that (1) $m$ overtook $t$, and (2) $j$ is not between $k$ and $i$, i.e. $\textit{furthest}_{k}(i,j)\neq i$. By (1), when $k$ received $m$, the condition in line 1 of RBM$_k$ would be true, so by line 2 of RBM$_k$ together with (2), $\textit{black}_{k}$ would be pushed beyond $i$. So at the last visit of $t$ to $k$, $\textit{black}_{t}$ was pushed beyond $i$. Then $\textit{black}_{i}$ would have been pushed beyond $i$ at line 2 of RT$_i$, contradicting the assumption that $\textit{black}_{i}=i$ at line 3 of RT$_i$. Suppose now that at line 3 of RT$_i$ the basic algorithm has not yet terminated. We argue that then $\textit{black}_{i}=i$ implies $\textit{count}_{t}>0$. There are two cases to consider. The first case is that some node is still active. Since nodes must be passive to forward $t$ and $i$ is passive (by $\textit{wait}(\textit{passive}_{i})$ in line 1 of RT$_i$), some node $k\neq i$ was made active after it forwarded $t$ for the last time, by a basic message $m$ from a node $j$ with $\textit{furthest}_{k}(i,j)=i$, while $t$ was traveling from $k$ to $j$. Since sending $m$ by $j$ is accounted for in $\textit{count}_{t}$ but its receipt at $k$ is not, by the time $t$ reaches $i$, $\textit{count}_{t}>0$. The second case is that a basic message is traveling from a node $j$ to a node $k$. Either sending this message was accounted for in $\textit{count}_{t}$, implying that its value at $i$ is positive, or $j\neq i$ became active after forwarding $t$ for the last time. Then it follows as in the first case that $\textit{count}_{t}>0$. 
\end{proof}
\section{Fault-Tolerant Version of Safra's Algorithm}
\label{sec:fault-tolerant}

From now on we assume nodes may spontaneously and permanently crash.
It is customary to assume for fault-tolerant distributed algorithms that there is a bidirectional channel between each pair of distinct nodes (see e.g.\ \cite{Tel00}), because else a node failure may result in disconnected subnetworks. Actually, it suffices if at any time a channel can be established between any two alive nodes.

Mittal et al.\ \cite{MFVP08} showed that a perfect failure detector is required to solve termination detection in the presence of failures.  In a fully asynchronous setting, such a detector cannot be built. As discussed at the end of Section \ref{sec:model}, a practical compromise is to assume an upper bound on the network latency. Each node sends out heartbeat messages at regular time intervals. When a node $i$ has not received a heartbeat from another node $j$ within some time interval, then $i$ permanently considers $j$ as crashed.

Each node $i$ stores the IDs of nodes it knows have crashed (either via its failure detector or from information in the token) in one of the sets $\textit{\sc Crashed}_{i}$ and $\textit{\sc Report}_{i}$. The latter contains the node IDs that $i$ has not yet reported to the other alive nodes by means of the token. When $i$ forwards the token, it adds all node IDs in $\textit{\sc Report}_i$ to a set $\textit{\sc Crashed}_t$ in the token and moves all those IDs from $\textit{\sc Report}_i$ to $\textit{\sc Crashed}_i$. Moreover, when $i$ receives the token, it removes all node IDs in $\textit{\sc Crashed}_i$ from $\textit{\sc Crashed}_{t}$, to avoid that it announces the same crashed node multiple times, and adds all node IDs in $\textit{\sc Crashed}_t$ to $\textit{\sc Crashed}_i$.

Since counts of messages to and from crashed nodes need to be discarded, the token contains $N$ counters, one per node; moreover, each node needs to count its message exchange with each other node separately and from the start of the execution run (instead of since the last token visit). So we split the field $\textit{count}_{i}$ for each node $i$ into a sequence $[\textit{count}_{i}^{0},\ldots,\textit{count}_{i}^{N-1}]$. For each node $j$, the field $\textit{count}_{i}^{j}$ stores the number of basic messages $i$ has sent to $j$ minus the number of basic messages $i$ has received from $j$. (The fields $\textit{count}_{i}^{i}$ are redundant as they always carry the value zero.)  If (the failure detector of) $i$ detects that a node $j$ has crashed, then $i$ permanently disregards the value of $\textit{count}_{i}^{j}$.
Likewise, to separately keep track of the counters at the different nodes in the token, the field $\textit{count}_{t}$ is split into a sequence $[\textit{count}_{t}^{\,0},\ldots,\textit{count}_{t}^{N-1}]$. If these counters were lumped together into a single counter $\textit{count}_{t}$, and say a node $i$ sent a basic message to a node $j$ which then crashed, there might be no way of telling whether or not $j$ received this message and updated $\textit{count}_{j}^{i}$ and $\textit{count}_{t}$.

As explained before, if a node $i$ learns from its failure detector that some other node $j$ crashed, it shares this information with the other alive nodes via the token. This is important because else there would be the risk that although $i$ from now on disregards $\textit{count}_{i}^{j}$, some other alive node $k$ may still take into account $\textit{count}_{k}^{j}$, which could lead to a premature termination detection at $k$.

Each node $i$ keeps track of its successor $\textit{next}_{i}$ in the ring; 
initially $(i+1)\bmod N$. Each time $i$ detects $\textit{next}_{i}$ has crashed, the value of this field is changed into $i$'s nearest alive successor.
We must ensure that the token is not lost; this could happen if the token was traveling to or being handled by $\textit{next}_{i}$ at the moment it crashed. Therefore, after having determined its new successor ${\textit{next}_{i}}'$,
$i$ forwards the token once again, to ${\textit{next}_{i}}'$. For this purpose $i$ stores 
the last token it forwarded. These local variables are updated as soon another token (with a higher sequence number) arrives.

In case $\textit{next}_{i}$ forwarded the token before crashing, ${\textit{next}_{i}}'$ could in this round receive a token both from $i$ and from $\textit{next}_i$. Therefore the token has a sequence number $\textit{seq}_{t}$, which is increased by one at each consecutive round trip of the token. In the first round $\textit{seq}_{t} = 1$. Each node $i$ keeps track of the highest sequence number it has passed on so far in $\textit{seq}_{i}$ (initially $\textit{seq}_{i} = 0$), and ignores incoming tokens with $\textit{seq}_{t} \leq \textit{seq}_{i}$. The last node in the ring, initially $N-1$, increases the sequence number every time it forwards the token. If it crashes, this task is taken over by its predecessor. A node $i$ can determine whether it is the last node by checking if $\textit{next}_{i}<i$.

As in the failure-sensitive variant of Safra's algorithm, $\textit{black}_{t}$ and $\textit{black}_{i}$ express which nodes are considered black, and when the token is sent by the initiator for the first time, $\textit{black}_{t}=N-1$ to guarantee it visits all nodes at least once. If a node receives an offending basic message, it colors all nodes in the ring  between itself and the sender black. If the failure detector of a node $i$ reports a crashed node and, at the next token visit, $i$ does not detect termination, then $i$ colors all other nodes black, as they must all be visited by the token.

The pseudocode of the procedures
is given in Algorithms \hyperref[algo:fts]{5-10}.
Again, the procedures should be executed without interruption, except that while waiting to become passive, in line 2 of ReceiveToken or line 5 of NewSuccessor, a node may perform SendBasicMessage, ReceiveBasicMessage, and FailureDetector calls.

In the initialization phase, nodes provide their local variables with initial values; node $0$ holds the token. Each node stores the values of the initial token in local variables with the subscript \textit{tBackup}.

\begin{minipage}{0.85\textwidth}
\setcounter{algocf}{4}
\begin{algorithm}[H]
    \label{algo:fts}
    \caption{$\mbox{\rm Initialization}_{i}$~}
    \For{$j=0$\, \mbox{\bf to} $N-1$}
    {
    $\textit{count}_{i}^{j} \leftarrow 0$;
    ~$\textit{count}_{tBackup}^{j} \leftarrow 0$\;
    }
    $\textit{black}_{i} \leftarrow i$; 
    ~~$\textit{seq}_{i} \leftarrow 0$;
    ~~$\textit{next}_{i} \leftarrow (i+1) \bmod N$\;
    $\textit{black}_{tBackup} \leftarrow i$; 
    ~~$\textit{seq}_{tBackup} \leftarrow 0$;
    ~~$\textit{\sc Crashed}_{tBackup} \leftarrow \emptyset$\;
    $\textit{\sc Crashed}_{i}\leftarrow\emptyset$;
    ~$\textit{\sc Report}_{i}\leftarrow\emptyset$\;
    \If{$i=0$\,}
    {
    \For{$j=0$\, \mbox{\bf to} $N-1$}
    {
    ~$\textit{count}_{t}^{j} \leftarrow 0$\;
    }
    $\textit{black}_{t} \leftarrow N-1$;
    ~~$\textit{seq}_{t} \leftarrow 1$;
    ~$\textit{\sc Crashed}_{t}\leftarrow\emptyset$\;
    ReceiveToken$_{0}$\;
    }
\end{algorithm}
\end{minipage}

At sending or receiving a basic message to or from an alive node, the sender or receiver updates the corresponding counter.

\begin{minipage}{0.85\textwidth}
\begin{algorithm}[H]
    \caption{$\mbox{\rm SendBasicMessage}_{i}(m,j)$~}
    \If{$j \notin \textit{\sc Crashed}_{i}\cup\textit{\sc Report}_{i}\cup\textit{\sc Crashed}_{t}$}{
        $\textit{seq}_{m} \leftarrow \textit{seq}_{i}$;
        ~~$\textit{send}(m,j)$;
        ~~$\textit{count}_{i}^{j} \leftarrow \textit{count}_{i}^{j} + 1$\;
    }
\end{algorithm}
\end{minipage}

Basic messages received from a crashed node in $\textit{\sc Report}_{i}$ may still be accounted for by $i$ in the control algorithm, to allow for termination detection at the next token visit to $i$. If the receipt of a message is accounted for in the token before its send, the receiver colors the nodes up to the sender black.

\begin{minipage}{0.85\textwidth}
\begin{algorithm}[H]
    \caption{$\mbox{\rm ReceiveBasicMessage}_{i}(m,j)$~}
    \If{$j \notin \textit{\sc Crashed}_{i}$}
    {
    \If{$\textit{seq}_{m} = \textit{seq}_{i}+1
    \vee (j > i \wedge\,\textit{seq}_{m} = \textit{seq}_{i})$}
        {$\textit{black}_{i}\leftarrow\textit{furthest}_{i}(\textit{black}_{i},j)$\;}
    $\textit{count}_{i}^{j} \leftarrow  \textit{count}_{i}^{j} - 1$\;
    }
\end{algorithm}
\end{minipage}
    
Procedure ReceiveToken$_{i}$ (RT$_i$) is executed when a token arrives at node $i$. It only proceeds if $i$ did not receive an instance of this token before (line 1). It then waits until it becomes passive,
because in the meantime the values of $\textit{count}_{i}^{j}$, $\textit{black}_{i}$ and $\textit{\sc Report}_{i}$ may still change. Once passive,
$\textit{black}_{i}$ is set to the furthest of $\textit{black}_{i}$ and $\textit{black}_{t}$ (line 2). Then, the set $\textit{\sc Crashed}_{t}$ is relieved of the nodes that $i$ reported through the token before (line 3). The remaining nodes in $\textit{\sc Crashed}_{t}$ are copied to $\textit{\sc Crashed}_{i}$, because they will be reported when $i$ forwards $t$ (line 4). $\textit{\sc Report}_{i}$ is relieved of nodes in $\textit{\sc Crashed}_{t}$ (line 5). The  values $\textit{count}_{i}^{j}$ for nodes $j\notin\textit{\sc Crashed}_{i}$ are accumulated in $\textit{count}_{t}^{i}$ (lines 7-9); but only if $i$ is white or $\textit{\sc Report}_{i}$ is empty (line 6), because then it may be employed in termination detection at $i$ (in lines 10-15) or at other nodes, respectively. If $i$ is white (line 10), the values $\textit{count}_{t}^{j}$ for nodes $j\notin\textit{\sc Crashed}_{i}$ are accumulated in an auxiliary field $\textit{sum}_{i}$ (lines 11-13); if this sum is $0$, $i$ announces termination (lines 14-15). If no termination is detected, then $i$ checks whether its successor is in $\textit{\sc Crashed}_{t}$; if so, NewSuccessor$_{i}$ is called to select another successor (lines 16-17). Next, $i$ checks whether it is the last node in the ring, and if so increases the sequence number of $t$ by $1$ (lines 18-19). If $\textit{\sc Report}_{i}$ is nonempty (line 20), then it is added to $\textit{\sc Crashed}_{t}$, so that $t$ will report these crashed nodes to all alive nodes; $\textit{black}_{t}$ is set to $i$, to ensure that the token visits all nodes up to $i$ again before termination can be detected, as all alive nodes must first achieve a consistent view on the set of crashed nodes (line 21). Next all nodes in $\textit{\sc Report}_{i}$ are moved to $\textit{\sc Crashed}_{i}$ (line 22). If $\textit{\sc Report}_{i}$ is empty, then $\textit{black}_{t}$ is set to $\textit{furthest}_{i}(\textit{black}_{i},\textit{next}_{i})$ (lines 23-24). Finally, $i$ forwards $t$ to $\textit{next}_{i}$ (line 25), makes copies of the token variables (lines 26-28), colors itself white, and increases $\textit{seq}_{i}$ by one (line 29).

\begin{minipage}{0.85\textwidth}
\begin{algorithm}[H]
    \caption{$\mbox{\rm ReceiveToken}_{i} (t)$~}
    \If{$\textit{seq}_{t}=\textit{seq}_{i}+1$}
    {
    $\textit{wait}\;(\textit{passive}_{i})$;
    ~~$\textit{black}_{i}\leftarrow \textit{furthest}_{i}(\textit{black}_{i},\textit{black}_{t})$\;
    $\textit{\sc Crashed}_{t}\leftarrow\textit{\sc Crashed}_{t}\setminus\textit{\sc Crashed}_{i}$\;
    $\textit{\sc Crashed}_{i}\leftarrow\textit{\sc Crashed}_{i}\cup\textit{\sc Crashed}_{t}$\;
    $\textit{\sc Report}_{i}\leftarrow\textit{\sc Report}_{i}\setminus\textit{\sc Crashed}_{t}$\;
    \If{$\textit{black}_{i}=i\,\vee\,\textit{\sc Report}_{i}=\emptyset$}
    {
    $\textit{count}_{t}^{i}\leftarrow 0$\;
         \For{{\rm all} $j\in\{0,\ldots,N{-}1\}\setminus\textit{\sc Crashed}_{i}$}
            {$\textit{count}_{t}^{i}\leftarrow \textit{count}_{t}^{i}+\textit{count}_{i}^{j}$\;}
    }
    \If{$\textit{black}_{i}=i$}
        {$\textit{sum}_{i}\leftarrow 0$\;
         \For{{\rm all} $j\in\{0,\ldots,N{-}1\}\setminus\textit{\sc Crashed}_{i}$}
            {$\textit{sum}_{i}\leftarrow \textit{sum}_{i}+\textit{count}_{t}^{j}$\;}
        \If{$\textit{sum}_{i}=0$}
         {{\rm Announce termination}\;}}
    \If{$\textit{next}_{i}\in \textit{\sc Crashed}_{t}$}
    {NewSuccessor$_{i}$\;}
    \If{$\textit{next}_{i}<i$}
    {$\textit{seq}_{t}\leftarrow\textit{seq}_{t}+1$\;}
    \eIf{$\textit{\sc Report}_{i}\neq\emptyset$}
    {
    $\textit{\sc Crashed}_{t}\leftarrow\textit{\sc Crashed}_{t}\cup\textit{\sc Report}_{i}$; ~~$\textit{black}_{t}\leftarrow i$\;
    $\textit{\sc Crashed}_{i}\leftarrow\textit{\sc Crashed}_{i}\cup\textit{\sc Report}_{i}$;
    ~~$\textit{\sc Report}_{i}\leftarrow\emptyset$\;
    }
    {
    $\textit{black}_{t}\leftarrow\textit{furthest}_{i}(\textit{black}_{i},\textit{next}_{i})$\;
    }
    $\textit{send}(t, \textit{next}_{i})$\;
    $\textit{\sc Crashed}_{\textit{tBackup}}\leftarrow\textit{\sc Crashed}_t$;
    ~~$\textit{black}_{\textit{tBackup}}\leftarrow\textit{black}_t$;
    ~~$\textit{seq}_{\textit{tBackup}}\leftarrow\textit{seq}_t$\;
    \For{$j=0$\, \mbox{\bf to} $N-1$}
    {
    ~$\textit{count}_{\textit{tBackup}}^{j} \leftarrow \textit{count}_t^j$\;
    }
    ~~$\textit{black}_{i} \leftarrow i$;
    ~~$\textit{seq}_{i}\leftarrow\textit{seq}_{i}+1$\;
    }
\end{algorithm}
\end{minipage}

\begin{minipage}{0.85\textwidth}
    \begin{algorithm}[H]
        \caption{$\mbox{\rm FailureDetector}_{i}$~}
        $\textit{crashed}(j)$\;
        \If{$j\notin\textit{\sc Crashed}_{i}\cup\textit{\sc Report}_{i}$}
        {
        $\textit{\sc Report}_{i}\leftarrow\textit{\sc Report}_{i}\cup\{j\}$\;
        \If{$j=\textit{next}_{i}$}
        {NewSuccessor$_{i}$\;
        \If{$\textit{seq}_{i}>0\vee\textit{next}_{i}<i$}
        {
            $\textit{\sc Crashed}_{\textit{tBackup}}\leftarrow \textit{\sc Crashed}_{\textit{tBackup}}\cup\textit{\sc Report}_{i}$\;
            $\textit{black}_{\textit{tBackup}}\leftarrow i$\;
        \If{$\textit{next}_{i}<i$}
        {$\textit{seq}_{\textit{tBackup}}\leftarrow\textit{seq}_{i}+1$\;}
        $\textit{send}(\textit{tBackup},\textit{next}_{i})$\;
        }
        }
        }
    \end{algorithm}
 \end{minipage}

FailureDetector$_{i}$ (FD$_i$) is invoked if $i$'s failure detector  reports that a node $j$ crashed ($\textit{crashed}(j)$ in line 1). If $i$ was not yet aware of this crash (line 2), then $j$ is added to $\textit{\sc Report}_{i}$ (line 3), so that this crash will be reported to other nodes via the token. If $j$ is the successor of $i$ in the ring, NewSuccessor$_{i}$ is invoked to compute a new successor of $i$  (lines 4-5). A backup token \textit{tBackup} is sent to the new successor (line 11), if $i$ received the token at least once (first disjunct in line 6); the second disjunct in line 6 ensures a backup token is sent when the initiator crashes before ever becoming passive. $\textit{\sc Report}_{i}$ is added to $\textit{\sc Crashed}_{\textit{tBackup}}$ (line 7); nodes in $\textit{\sc Report}_{i}$ are not transposed to $\textit{\sc Crashed}_{i}$ yet, because the backup token may be discarded in favor of the original token. By $\textit{black}_{\textit{tBackup}} \leftarrow i$ (in line 8) it is guaranteed that the backup token visits all nodes up to $i$ again before termination can be detected, as all alive nodes must take into account the crash of $j$. If no alive node has an identity greater than $i$, then $\textit{seq}_{\textit{tBackup}}$ is increased by one (lines 9-10).

NewSuccessor$_{i}$ (NS$_i$) computes $i$'s new successor after $\textit{next}_{i}$ crashed. First, $\textit{next}_{i}$ is changed into $(\textit{next}_{i}+1)\bmod N$ (line 1). Then, it is repeatedly checked whether the new value of $\textit{next}_{i}$ is a crashed node (line 2), and if so its value is increased by one, modulo $N$ (line 3). After the value of $\textit{next}_{i}$ has stabilized, $i$ checks whether it is the only remaining alive node in the network (line 4), and if so, waits until it has become passive to announce termination (lines 5,7). Else, if $\textit{black}_{i}\neq i$, then $\textit{black}_{i}$ is set to $\textit{furthest}_{i}(\textit{black}_{i},\textit{next}_{i})$
(lines 8-9). Line 6, updating $\textit{\sc Crashed}_i$ right before the announcement, is new in comparison with the algorithmic description in \cite{KFF21}, and will be commented on in Section \ref{sec:mcrl2}.

 \begin{minipage}{0.85\textwidth}
    \begin{algorithm}[H]
        \caption{$\mbox{\rm NewSuccessor}_{i}$~}{}{}
        $\textit{next}_{i}\leftarrow (\textit{next}_{i}+1)\bmod N$\;
        \While{$\textit{next}_{i}\in\textit{\sc Crashed}_{i}\cup\textit{\sc Report}_{i}$}
        {$\textit{next}_{i}\leftarrow(\textit{next}_{i}+1)\bmod N$\;}
        \If{$\textit{next}_{i}=i$}
        {
        $\textit{wait}(\textit{passive}_{i})$\;
        $\textit{\sc Crashed}_i\leftarrow \textit{\sc Crashed}_i\cup\textit{\sc Report}_i$;
        ~~$\textit{\sc Report}_i\leftarrow\emptyset$\;
        {\rm Announce termination}\;
        }
        \If{$\textit{black}_{i}\neq i$}
        {
        $\textit{black}_{i}\leftarrow\textit{furthest}_{i}(\textit{black}_{i},\textit{next}_{i})$\;
        }
    \end{algorithm}
 \end{minipage}

\begin{figure*}[ht]
\begin{subfigure}{\textwidth}
    \centering
    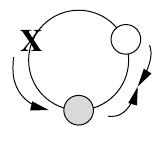
    \subcaption[]{Node 0 sends two messages and forwards the token before crashing}
    \label{ft-a}
\end{subfigure}

\vspace{4mm}

\begin{subfigure}{\textwidth}
    \centering
    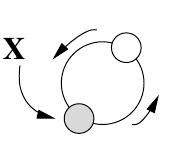
    \subcaption[]{After the crash is detected, the two messages are discarded.}
    \label{ft-b}
\end{subfigure}

\vspace{4mm}

\begin{subfigure}{\textwidth}
    \centering
    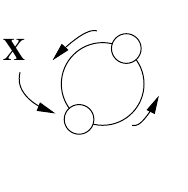
    \subcaption[]{Termination is detected in the next round}
    \label{ft-c}
\end{subfigure}
\caption{Example run on a faulty network of three nodes}
\label{example:ft}
\end{figure*}

\begin{example}
We consider one possible run of our fault-tolerant algorithm on a ring of three nodes in Figure \ref{example:ft}.
Initially all nodes are active, all counters carry the value $0$, and $\textit{black}_{i}=i$ and $\textit{seq}_{i}=0$ for $i=0,1,2$. Node $0$ sends basic messages $m$ and $m'$ to node $1$, node $1$ sends basic message $m''$ to node $2$, and node $2$ sends basic message $m'''$ to node $1$ (all with their node ID and sequence number $0$ attached); $\textit{count}_{0}^{1}$ is set to $2$, and $\textit{count}_{1}^{2}$ and $\textit{count}_{2}^{1}$ are set to $1$. Nodes $0$ and $2$ now become passive. Node $0$ sends the token to node $1$ (with $\textit{count}_{t}^{0}=2$, $\textit{count}_{t}^{1}=\textit{count}_{t}^{2}=0$, $\textit{black}_{t}=2$, $\textit{seq}_{t}=1$ and $\textit{\sc Crashed}_{t}=\emptyset$), and crashes. 
This leads to Figure \ref{ft-a} where the cross at node $0$ represents that it has crashed, the sequences of {\it count} values at alive nodes are placed between square brackets, and empty $\textit{\sc Crashed}$ and $\textit{\sc Report}$ sets at nodes have been omitted.

In  Figure \ref{ft-b}, node $2$ detects node $0$ crashed and sets $\textit{\sc Report}_{2}$ to $\{0\}$; from now on node $2$ ignores $\textit{count}_{2}^{\,0}$. Since $\textit{next}_{2}=0$, node $2$ makes node $1$ its new successor. Since $1<2$, node $2$ sends a backup token to node $1$ (with $\textit{count}_{t}^{1}=\textit{count}_{t}^{2}=0$, $\textit{black}_{t}=2$, $\textit{seq}_{t}=1$ and $\textit{\sc Crashed}_{t}=\{0\}$). Node $1$ receives $m$ from node $0$ and sets $\textit{count}_{1}^{0}$ to $-1$; since the sender of $m$ is $0<1$ and $\textit{seq}_{m}=0=\textit{seq}_{1}$, $\textit{black}_{1}$ remains unchanged; moreover, node $1$ receives $m'''$ from node $2$ and sets $\textit{count}_{1}^{2}$ to $0$; since the sender of $m'''$ is $2>1$ and $\textit{seq}_{m'''}=0=\textit{seq}_{1}$, $\textit{black}_{1}$ is set to $2$. 
Then, in  Figure  \ref{ft-c}, node $1$ receives the backup token from node $2$ and sets $\textit{\sc Crashed}_{1}$ to $\{0\}$. It becomes passive, passes on the token to node $2$ with $\textit{count}_{t}^{1}$ set to $\textit{count}_{1}^{2}=0$, and sets both $\textit{black}_{1}$ and $\textit{seq}_{1}$ to $1$. Node $1$ does not detect termination since $\textit{black}_{t}=2$. Next, node $1$ receives the original token from node $0$, which is dismissed. Node $2$ receives the token and sets $\textit{\sc Crashed}_{2}$ to $\{0\}$ and $\textit{\sc Report}_{2}$ to $\emptyset$. It does not detect termination because it sets $\textit{count}_{t}^{2}$ to $\textit{count}_{2}^{1}=1$, and $\textit{count}_{t}^{1}+\textit{count}_{t}^{2}=0+1>0$. It passes on the token to node $1$ with $\textit{black}_{t}=1$ and $\textit{seq}_{t}=2$, and sets $\textit{black}_{2}$ to $2$ and $\textit{seq}_{2}$ to $1$.

When the token arrives, node $1$ sets $\textit{\sc Crashed}_{t}$ to $\emptyset$, computes $\textit{count}_{t}^{1}=1$, passes on the token to node $2$ with $\textit{black}_{t}=2$, and sets $\textit{black}_{1}$ to $1$ and $\textit{seq}_{1}$ to $2$. In the meantime node $2$ receives $m''$ from node $1$ and sets $\textit{count}_{2}^{1}$ to $0$; since the sender of $m''$ is $1<2$ and $\textit{seq}_{m''}=0<\textit{seq}_{2}$, $\textit{black}_{2}$ remains unchanged. Node $2$ becomes passive again. When the token arrives, node $2$ computes $\textit{count}_{t}^{1}+\textit{count}_{t}^{2}=0+0=0$. Since also $\textit{black}_{t}=2$, it announces termination.\ Finally node $1$ ignores message $m'$ from node $0$, because $0\in\textit{\sc Crashed}_{1}$.

\end{example}

\begin{theorem}
The fault-tolerant version of Safra's termination detection algorithm is correct.
\end{theorem}
\begin{proof}
We explain how the line of reasoning in the proof of Theorem 1 needs to be extended and adapted for the fault-tolerant version of Safra's algorithm.
Since failure detectors are (effectively) perfect, a node only ends up in a $\textit{\sc Crashed}$ set if has crashed. 
Counting basic messages between alive nodes must be performed in a consistent way at all alive nodes, so they must have a uniform view on their $\textit{\sc Crashed}$ sets. Therefore, when node $i$ reports newly detected crashed nodes in the token (line 21 of RT$_i$), it makes sure the token must visit all nodes up to $i$ again before termination can be detected (by $\textit{black}_{t} \leftarrow i$).
Furthermore, if $i$ issues a backup token (line 11 of FD$_i$) after its successor crashed, then it makes sure the token will complete a round trip (by line 8 of FD$_i$), so that this crash is accounted for at all alive nodes.

First we argue {\it Liveness}: if the basic algorithm has terminated and not all nodes in the network crash, then eventually termination will be announced by some node. We observe that the token continuously visits all consecutive alive nodes in the ring, if there are at least two of those, since NS$_i$ is called when the successor of node $i$ in the ring is found to have crashed (lines 16-17 of RT$_i$ and 4-5 of FD$_i$). Moreover, owing to the backup mechanism (lines 6-11 of FD$_i$) the token is not lost at a crash; and by managing $\textit{seq}_{t}$ and $\textit{seq}_{i}$ (in lines 18-19 and 29 of RT$_i$ and 9-10 of FD$_i$) we make sure that only a single token is being forwarded (by line 1 of RT$_i$). This means that the argumentation why the failure-sensitive version of Safra's algorithm announces termination after the basic algorithm has terminated also applies here. The only real difference is that this argumentation is restricted to the nodes that are still alive, and that termination is only announced if at least one node stays alive.

Now we argue {\it Safety}: If termination is announced, then (1) all alive nodes are passive, and (2) for all basic messages in the channels, the receiver either has crashed or knows that the sender has crashed. There are two places where a node $i$ can call ``Announce termination''. In case of line 7 of NS$_i$ we have $\textit{next}_{i}=i$ (by line 4 of NS$_i$), which implies that all other nodes have crashed. So when $i$ becomes passive (in line 5 of NS$_i$). Since moreover $\textit{\sc Crashed}_i$ is up-to-date (by line 6 of NS$_i$), ensuring that it cannot be made alive by a basic message from a crashed node, it can safely announce termination. In case of line 15 of RT$_i$ we have $\textit{sum}_{i}=0$ (by line 14 of RT$_i$), meaning that the counters of the active nodes add up to $0$ (by lines 7-9 and 11-13 of RT$_i$). With a similar line of reasoning as for the improved version of Safra's algorithm it can be argued that this implies the basic algorithm has terminated. 
\end{proof}

It can be argued that when an execution terminates, this is always detected within $\frac{N^2}{2}$ token steps. The worst case occurs if $N-1$ times a crash happens just before the token completes its round trip.
\section{Model Checking Analysis}
\label{sec:mcrl2}

We modelled the improved version of Safra's algorithm and our fault-tolerant version in mCRL2 \cite{BGKLNVWWW19} and performed a model checking analysis. In this analysis, the basic algorithm of which termination is being detected is abstracted away: Nodes can be active, passive, or crashed, and only the number of basic messages in transit in a channel is recorded. We put an (artificial) upper bound $M$ on the total number of messages sent by the basic algorithm, to achieve a finite state space for the model checking analysis. (Similarly, in the model checking analysis of Safra's algorithm in \cite{KKM22} a bound is put on the number of basic messages each node can send.) Moreover, we put a bound $S$ on the maximum number of times the token can go around the virtual token ring. Once this maximum is reached, we require that the token can only be forwarded when the basic algorithm has terminated, meaning that all alive nodes are passive and there are no more basic messages from alive nodes in transit.

As explained in Section \ref{sec:model}, a perfect failure detector is assumed which never falsely suspects that a node crashed and eventually detects each node crash. A node can crash nondeterministically, after which it no longer performs any activity. The information of such a crash is passed on to the alive nodes in an asynchronous manner, by a global failure detector process. (Alternatively, in the model checking analysis of Tseng's termination detection algorithm in \cite{HAN18} a crashed node informs the alive nodes that it has crashed.) To keep the state space at bay, nodes cannot crash after performing an internal event which affects only the local state of the node (such as updating a local variable or becoming passive/active), as opposed to external events that involve sending or receiving a basic message or announcing termination. This is so because whether a node crash occurs before or after an internal event does not make a difference for the overall system behavior.

We performed a model checking analysis by casting some desired system properties in the regular alternation-free $\mu$-calculus \cite{MateescuS03}, which is the temporal logic employed by the model checker of mCRL2. The two main properties we verified are a liveness and a safety property:
\begin{itemize}
\item[(1)]
If the execution of the basic algorithm has terminated, then termination will eventually be announced.
\item[(2)]
If termination is announced, then the execution of the basic algorithm has terminated.
\end{itemize}
We moreover verified some basic safety properties, such as that values of certain variables stay within their intended ranges.\footnote{The mCRL2 specifications and regular alternation-free $\mu$-calculus formulas underlying this model checking exercise can be found at \url{https://github.com/Andy-Tatman/FaultTolerantSafra-mCRL2-Artifact}.}

The model checking experiment was performed on a machine with a 7800x3D CPU (8 cores/16 threads) and 32GB RAM. Due to the relatively large number of parameters in the algorithm in combination with the nondeterministic behaviour of the basic algorithm, we could only model check the failure-sensitive and fault-tolerant versions of Safra's algorithm for $M=2$ and $S=2$, whereby the failure-sensitive version could be analyzed up to $N=4$ nodes and the fault-tolerant version up to $N=3$ nodes. To analyze safety properties for larger instances, we applied so-called highway search \cite{EGWW09}, whereby the state space is generated using breadth-first search, but only a limited number of states at each level is used for further exploration. By using different selection techniques or randomization for this further exploration, different searches explore different parts of the state space.

This analysis revealed one subtle flaw, in the fault-tolerant version from \cite{KFF21}, which can occur if the entire network except for one node crashes. Consider the following execution, with $N=2$. Node 1 sends a basic message $m$ to node 0 and then crashes. Node 0 becomes passive; since it holds the token, it continues with $\mbox{ReceiveToken}_0$ and sends the token to node 1. Next, node 0 detects node 1’s crash and sets $\textit{\sc Report}_0 = \{1\}$. Since $\textit{next}_0=1$, node 0 performs $\mbox{NewSuccessor}_0$, where it finds that the new value of $\textit{next}_0$ is 0 itself. Since node 0 is passive, it announces termination. Now basic message $m$ from node 1 arrives at node 0. Since 1 is not in $\textit{\sc Crashed}_0$ (but in $\textit{\sc Report}_0$), node 0 runs $\mbox{ReceiveBasicMessage}_0(m,1)$ and becomes active.

To avoid this erroneous scenario, it suffices to let node $i$ in $\mbox{NewSuccessor}_i$ update $\textit{\sc Crashed}_i$, by adding the nodes in $\textit{\sc Report}_i$ to it, right before announcing termination. For this reason, line 6 was added to $\mbox{NewSuccessor}_i$. After this adaptation, all the properties we checked in mCRL2 turned out to hold.
\section{Conclusion}
\label{sec:conclusion}

We presented a fault-tolerant algorithm for distributed termination detection based on an improved version of Safra's algorithm. In our fault-tolerant variant, message counters are maintained per node, so that counts to and from crashed nodes can be discarded. If a node crashes, the ring structure is restored locally and a backup token is sent. Strong points are: little message overhead when nodes remain active for a long time; robust against any number of and simultaneous node crashes; only one additional message per crash; the basic algorithm can be decentralized; no leader election scheme; no underflow issues. Compared to other algorithms, our algorithm generates far fewer, but larger control messages. For overall performance, fewer messages tend to be better, since more messages mean more processing at each node, as well as at the network stack.

Next to correctness proofs for the failure-sensitive and fault-tolerant versions of Safra's algorithm, we performed a model checking analysis. This analysis suffered from the Achilles heel of explicit-state model checking, state space explosion, which was circumvented by applying highway search. The analysis showed the strength of model checking in detecting a subtle flaws on small network configurations, by unveiling an erroneous corner case in the fault-tolerant version from \cite{KFF21}. We moreover showed how to fix this flaw.

Experiments in \cite{KFF21} indicate that our algorithm does not impose significant extra overhead in control messages compared to 
its failure-sensitive counterpart. Despite the $O(N)$ bit complexity of the token, the available throughput and low latency of current network technologies, as well as the low message complexity of our algorithm, may render our approach feasible for large networks. This would need to be validated in experiments with real-life distributed networks under realistic and diverse workloads on many machines. 

Since it is hard to get distributed algorithms correct, as witnessed by the subtle bug revealed by our model checking analysis, it would be useful to verify the results in this paper by putting the proofs in a proof checker like Rocq or Isabelle. Another interesting direction for future work is to develop a version of our fault-tolerant algorithm in the presence of stable storage. In that case the memory overhead of splitting the counter in the token can be avoided, at the cost of storing message counts in stable storage. 
\bibliographystyle{fundam}
\bibliography{references}

\end{document}